\documentclass[letterpaper, 10 pt, conference]{ieeeconf}  

\IEEEoverridecommandlockouts                              

\usepackage{xcolor}
\usepackage{amsmath}
\usepackage{amsfonts}
\usepackage{graphicx}
\usepackage{mathtools}
\usepackage{stmaryrd}
\usepackage[ruled,vlined]{algorithm2e}
\usepackage{caption}
\usepackage{subfigure}
\usepackage{tikz}
\usepackage{graphicx}
\usepackage{svg}
\usepackage{amsmath}
\usepackage{amssymb}
\usepackage{mathabx}
\usepackage{hyperref}
\usepackage{lipsum}
\usepackage{booktabs}
\hypersetup{
    colorlinks=true,
    linkcolor=black,
    urlcolor=blue,
    pdftitle={},
    pdfpagemode=FullScreen,
    }

\usetikzlibrary{shapes.geometric, arrows}
\usetikzlibrary{calc,fit,arrows}

\newtheorem{theorem}{Theorem}[section]

\newtheorem{assumption}[theorem]{Assumption}

\newtheorem{definition}[theorem]{Definition}

\newtheorem{remark}[theorem]{Remark}
\newtheorem{problem}[theorem]{Problem}

\newcommand{\Ze}{{\mathbb Z}}
\newcommand{\R}{{\mathbb{R}}}
\newcommand{\N}{{\mathbb{N}}}

\newcommand{\segcc}[1]{{{\left\llbracket#1\right\rrbracket}}}

\tikzstyle{robot} = [circle, text centered, draw=black]
\tikzstyle{module} = [rectangle, minimum width=2cm, minimum height=1cm, text centered, text width=2cm, draw=black]
\tikzstyle{arrow} = [thick,->,>=stealth]

\title{\LARGE \bf
Barrier Function-based Distributed Symbolic Controller for Multi-Agent Systems
}
\author{David Smith Sundarsingh$^{1,*}$, Ratnangshu Das$^{1,*}$, Adnane Saoud$^2$, and Pushpak Jagtap$^1$ 
\thanks{This work was supported in part by the Google Research Grant, the SERB Start-up Research Grant, and the CSR Grants by Siemens and Nokia. }
\thanks{$^*$The authors contributed equally.}
\thanks{$^1$D. S. Sundarsingh, R. Das, and P. Jagtap are with the Robert Bosch Centre for Cyber-Physical Systems, IISc, Bangalore, India {\tt\small \{davids,ratnangshud,pushpak\}@iisc.ac.in}}%
\thanks{$^2$A. Saoud is with the College of Computing, University Mohammed VI Polytechnic, Benguerir, Morocco {\tt\small adnane.saoud@um6p.ma}}
}

\usepackage{booktabs,siunitx}
\usepackage{tabularx,ragged2e}
\newcolumntype{L}{>{\RaggedRight\hangafter=1\hangindent=1em}X}

\begin{document}

\maketitle
\thispagestyle{empty}
\pagestyle{empty}

\begin{abstract}
Because of the scalability issues associated with the symbolic controller synthesis approach, employing it in a multi-agent system (MAS) framework becomes difficult. 
In this paper, we present a novel approach for synthesizing distributed symbolic controllers for MAS, that enforces a local Linear Temporal Logic (LTL) specification on each agent and global safety specifications on the MAS, in a computationally efficient manner by leveraging the concept of control barrier functions (CBF).
In addition, we also provide an analysis on the effect of the CBF parameters on the conservatism introduced by our proposed approach in the size and domain of the synthesized controller. 
The effectiveness of this approach is demonstrated through a comparison with the conventional monolithic symbolic control, using simulation as well as hardware demonstrations.

\end{abstract}

\section{Introduction}
{In recent years, MAS has emerged as a prominent field of research, encompassing various domains such as robotics, automation, and artificial intelligence. These systems consist of multiple autonomous agents coordinating their actions to achieve common objectives. When subjected to complex specifications, ensuring that each agent accomplishes its local task while adhering to global tasks is crucial. Consider the example where a swarm of autonomous drones is deployed to perform a search and rescue operation. Each drone is assigned a specific local task, like scanning designated regions or delivering essential supplies. However, it is imperative to ensure that the drones operate safely and do not pose a risk to each other or the people on the ground. In this case, the global specifications would include rules and constraints that govern the behavior of the drone swarm to prevent collisions, maintain desired distances, and avoid hazardous situations. The behavior and performance of MAS critically depend on the design and implementation of their controllers. Such complex tasks can efficiently be represented using linear temporal logic (LTL) \cite{tabuada2009verification}.
To synthesize controller enforcing LTL specifications, the symbolic control approach \cite{tabuada2009verification,ReissigWeberRungger17} is widely used.
However, these approaches often face challenges related to computational complexity, scalability, and the ability to handle large-scale MAS effectively, resulting in the issue of the so-called curse of dimensionality.}

Controller synthesis for MAS is usually done in two ways - top-down and bottom-up approaches. The top-down approach involves the decomposition of global tasks into local ones. For example, \cite{GAMARA} and \cite{autoSynth} decompose global controllers, and each agent solves for the decomposed strategy. The authors in \cite{MDP} modified this approach by only solving for a sub-group of the MAS, which has mutual specifications. \cite{Auction} and \cite{swarmreactive} assign tasks to each agent based on a global specification. These approaches assume that the global specification is decomposable. 
On the other hand, bottom-up approaches solve for local specifications with some constraints \cite{co-opCoupled} or synthesize controllers again after composing the system from individually controlled agents \cite{MITL,saoud2021compositional}. While these methods can solve local and global specifications, they require some assumptions on the system \cite{MITL,co-opCoupled} or on the specifications \cite{saoud2021compositional}. 
A combination of the two approaches \cite{decoupled}, decomposes the system, solves the local specification, and removes any conflicts after recomposing the system. This process is repeated iteratively with no guarantee on the convergence of the solution.

Controller synthesis for MAS using CBFs has gained significant attention in recent research \cite{multi-barrier}, especially for enforcing a fragment of signal temporal logic (STL) specifications.
The authors in \cite{safetyOnly} used barrier functions to synthesize a least-violating controller for collision avoidance. Additionally, the integration of CBFs within symbolic control techniques has also been explored. In \cite{acASR+barrier} and \cite{barrierabstract}, the controller computes a discrete plan and employs barrier functions to ensure safe transitions. 
Furthermore, \cite{verifiable} builds upon the work in \cite{safetyOnly} by employing a nominal controller synthesized in a centralized manner and projecting it onto each agent to fulfil LTL specifications in addition to collision avoidance. However, since this approach uses centralized controller synthesis, it encounters scalability challenges.

{To address these limitations, we presented a novel bottom-up symbolic controller synthesis approach for MAS by leveraging the concept of barrier certificates. 
The controller synthesis is done in three steps: synthesizing local controllers for individual agents using symbolic control techniques, composing these controllers using barrier certificates for safety specification, and then synthesizing a controller for the resulting multi-agent system {to ensure both local LTL specifications and global safety specification.}
While the proposed approach significantly enhances computational efficiency compared to conventional centralized techniques, it introduces a degree of conservatism in terms of the size and domain of the controller. The preliminary results are presented in \cite{sundarsingh2023safeMulti}. In this paper, we build on the work by providing formal guarantees and a thorough analysis on the effect of the barrier certificate parameters on conservatism and computational efforts. Furthermore, a comparative study between the controller synthesized by our approach and the conventional monolithic controller synthesis illustrates the significant reduction in the computation time.
Finally, we conduct hardware studies to showcase the practicality of the proposed approach using a MAS setup with local reach-avoid and global safety specifications. 

\section{Preliminaries and Problem Definition}
\textbf{Notations:}
For $x \in \mathbb{R}^n$, $x_q$ represents the $q^{th}$ element of the vector $x$, where $q\in\{1,\dots,n\}$ and the infinity norm of $x$ is $\lVert x\rVert:=\mathbf{max}_{q\in\{1,\dots n\}}|x_{q}|$. For $a,b\in(\R\cup\{-\infty, \infty\})^n$, where $a\leq b$ component-wise, the closed hyper-interval is denoted by $\segcc{a,b}:=\R^n\cap([a_1,b_1]\times\dots\times[a_n,b_n])$. A relation $R\subseteq A\times B$ can be defined as a map $R:A\rightarrow 2^B$ as follows: $b\in R(a)$ iff $(a,b)\in R$.  The relation $R$ is strict if $R(a)\neq\emptyset$, $\forall a\in A$. The inverse of the relation is defined as $R^{-1}:=\{(b,a)\in B\times A|(a,b)\in R\}$ and can be written as $a\in R^{-1}(b)$. Given a set $S$, $Int(S)$ and $\partial S$ represent the interior and the boundary of $S$, respectively. 
Consider $N$ sets $A_i$, $i\in\{1,\ldots,N\}$, the Cartesian product of sets is given by $A=\prod_{i\in\{1,\ldots,N\}}A_i:=\{(a_1,\ldots,a_N)|a_i\in A_i,i\in\{1,\ldots,N\}\}$. Given $N$ functions $f_i:X_i\rightarrow A_i$, the Cartesian product of functions is $f:X\rightarrow A:=\prod_{i\in\{1,\ldots,N\}}f_i=(f_1(x_1),\ldots, f_N(x_N))$.
The composition of two maps $H$ and $R$ is $H\circ R(x):=H(R(x))$. A function $\alpha: \R_0^+ \rightarrow \R_0^+$ is of class $\mathcal{K}$ if $\alpha(0)=0$, and it is strictly increasing. If $\alpha \in \mathcal{K}$ is unbounded, it is of class $\mathcal{K}_{\infty}$.

\subsection{Discrete-time Multi-Agent Systems}

{Consider a collection of $N\in \N$ agents and let $I=\{1,\dots,N\}$. Each agent's state evolution is given by the following discrete-time control system: 
\begin{align}
\mathsf{x}_i(k+1)= \mathsf{f}_i(\mathsf{x}_i(k),\mathsf{u}_i(k)),\text{ } i\in I,\text{ }k\in\N_0,
\label{sys}
\end{align}
where ${\mathsf{x}_i}(k)\in X_i\subset \R^{n_i}$ 
is the state of the $i^{th}$ agent and ${\mathsf{u}_i}(k)\in U_i\subset\R^{m_i}$ is the input to the agent.\\
The evolution of the multi-agent system (MAS) is given by:
\begin{align}
{\mathsf{x}}(k+1)= {\mathsf{f}(\mathsf{x}}(k),{\mathsf{u}}(k)),\text{ }k\in\N_0, \text{ }x(0) \in X^0,
\label{sysMAS}
\end{align}
where $X^0$ is the set of initial states, $\mathsf{x}(k)\in X:=\prod_{i\in I}X_i\subset\R^n$ is the state of the multi-agent system, $n=\Sigma_{i\in I}n_i$, $\mathsf{u}(k)\in U:=\prod_{i\in I}U_i\subset\R^m$ is the input to the system and $m=\Sigma_{i\in I}m_i$. The function $\mathsf{f}$ is given by ${\mathsf{f}}:X\times U\rightarrow X$ and ${\mathsf{f(x}}(k),{\mathsf{u}}(k)):=\prod_{i\in I}{\mathsf{f_i(x_i}}(k),{\mathsf{u_i}}(k))$.
The trajectory of system \eqref{sysMAS} starting from a state $x\in X^0$ under the input signal $\mathsf{u}$ is given by $\mathsf{x}_{x\mathsf{u}}$ and $\mathsf{x}_{x\mathsf{u}}(k)$ gives the value of the state of the system {along that trajectory} at sampling instance $k$.


 The reachable set of system \eqref{sysMAS} from a set $\mathcal{X} \subseteq X$ 
under an input $u\in U$ 
is given by $Reach(\mathcal{X},u):=\bigcup_{x\in \mathcal{X}}\mathsf{f}(x,u)$, 
which is the set of all states that the system "reaches" when an input $u$ is applied at all the states in $\mathcal{X}$ in a one-time step. 
This reachable set is difficult to compute, so we use the over-approximated reachable set, $\overline{Reach}(x,u)$. Several approaches are available in the literature for computing this over-approximated set; for example, \cite{reachable1,reachIntro} and \cite{ReissigWeberRungger17}. 

\subsection{Transition Systems}
We now introduce the notion of transition systems \cite{tabuada2009verification}
which will serve as a unified representation for discrete-time control systems and their corresponding symbolic models. 

\begin{definition}
\label{Def:trans_sys}
A transition system is a tuple $\Sigma=(X,X^0,U,F)$, where 
$X$ is the set of states, $X^0\subseteq X$ is the set of initial states, $U$ is the set of inputs, and the map $F:X\times U\rightrightarrows X$ is the transition relation. 
\end{definition}
The set of admissible inputs for $x\in X$ is denoted by { $U^a(x):=\left\{u \in U \mid F\left(x, u\right) \neq \emptyset\right\}$}. We use $x^{\prime} \in F\left(x, u\right)$ to represent the $u$-successor of state $x$.\\
Consider the discrete-time control system of agent $i$ as given in \eqref{sys}. The transition system representation of the $i^{th}$ agent is given by the tuple $\Sigma_i=(X_i,X^0_i,U_i,F_i)$, where $X_i\subset\R^{n_i}$ is the set of the states of agent $i$, $X_i^0\subseteq X_i$ is the set of initial states, $U_i\subset\R^{m_i}$ is the set of inputs for agent $i$ and for $x_i\in X_i$, $u_i\in U_i$, $F_i(x_i,u_i):=\mathsf{f}_i(x_i,u_i)$.\\
The composition of the $N$ transition systems, which represents the multi-agent system \eqref{sysMAS} is given by definition below.
\begin{definition}\label{DefComposition}
		Given a collection of $N\in\N$ agents represented as  $\left\{\Sigma_{i}\right\}_{i \in I}$, where
        $I=\{1,\dots,N\}$, the composed transition system is $\Sigma=\left(X,X^0,U,F\right)$, where
		\begin{itemize}\raggedright
			\item $X=\prod_{i \in I} X_{i}$, $X^0\subseteq X$, $U=\prod_{i \in I} U_{i}$,
			\item for $x\in X$ and $u\in U$, $F(x,u)=\prod_{i\in I}F_{i}\left(x_{i}, {u}_{i}\right)$.
		\end{itemize}
\end{definition}
 
\subsection{Problem Formulation}

Next, we formally state the problem considered here:

\begin{problem}\label{problem}
Given $N$ agents with dynamics as in \eqref{sys}, local specifications $\psi_i$ expressed as LTL formulae \cite{LTLcite}
for each agent and global safety specification $\Phi$ over the MAS \eqref{sysMAS}; design a controller for the MAS that enforces global safety specification and the local LTL specifications.
\end{problem}

One can solve Problem \ref{problem} for a MAS monolithically using symbolic control techniques (as discussed in Section III) for satisfying the global and local specifications $\varphi=\Psi\wedge\Phi$, where $\Psi = \bigwedge_{i\in I}\psi_i$,  but scalability and computational efficiency are stringent issues for controller synthesis.

To deal with the scalability issue and, as a consequence, to reduce computational complexity, we propose to use a bottom-up approach consisting of three steps:
\begin{enumerate}
    \item We construct a symbolic controller for each of the $N$ agents $\Sigma_i$ to satisfy the corresponding local LTL specification $\psi_i$, where $i \in \{1,\ldots,N\}$, as explained in Section \ref{symbolCont}. Using the symbolic controller, we construct the individual controlled agents.
    \item For the composed controlled MAS, resulting from the composition of the controlled agents obtained in step $(1)$, we construct a barrier function \cite{barrierfunc} that enforces the global safety specification $\Phi$. We use a notion of barrier functions (introduced in Section \ref{barriersection}) as certificates to remove transitions that violate the safety specification in the composed system.
    \item Since the local specifications may be violated due to the safety-enforcing barrier functions, we synthesize a controller for the controlled system obtained in step $(2)$ to achieve the specification $\Psi = \bigwedge_{i\in I}\psi_i$.
\end{enumerate}
The proposed controller ensures that the MAS satisfies the local specification $\Psi$ and the global specification $\Phi$. 

\section{Symbolic Control}\label{symbolCont}
In this section, we will briefly discuss symbolic control and how it can be used to synthesize controllers to satisfy a given specification.
In order to relate the discrete-time control system and its symbolic model, we use the notion of feedback refinement relation \cite{ReissigWeberRungger17}.
\subsection{Feedback Refinement Relation}
\begin{definition}\label{FRRdef}
Consider two transition systems $\Sigma=(X,X^0,U,F)$ and $\hat{\Sigma}=(\hat{X},\hat{X}^0,\hat{U},\hat{F})$. A strict relation $Q\subseteq X \times \hat{X}$ is said to be a feedback refinement relation from $\Sigma$ to $\hat{\Sigma}$, denoted by $\Sigma\preceq_Q \hat{\Sigma}$, if for each $(x, \hat{x}) \in Q$ the following conditions hold:
\begin{itemize}
    \item $\hat{U}^a(\hat{x}) \subseteq U^a(x)$,
    \item $u \in \hat{U}^a(\hat{x})\implies Q(F(x, u))\subseteq \hat{F}(\hat{x}, u)$.
\end{itemize}
\end{definition}
The feedback refinement relation $Q$ allows to transform a controller for the symbolic model $\hat{\Sigma}$ into a controller for the original system $\Sigma$ {as any input available for $\hat{\Sigma}$ is available for $\Sigma$ and the transitions of $\hat{\Sigma}$ preserve behavioural relations with the corresponding transitions in $\Sigma$.}
\subsection{Construction of Symbolic Models}\label{symconst}

In order to synthesize controllers for the concrete system $\Sigma$, we need to construct its symbolic model $\hat{\Sigma}$, which is related to $\Sigma$ through the feedback refinement relation.

\begin{definition}\label{Def:SymbolicModel}
The symbolic model of the system $\Sigma$ is given by $\hat{\Sigma}=(\hat{X},\hat{X}^{0},\hat{U},\hat{F})$, where
\begin{itemize}
    \item $\hat{X}$ is a cover over $X$ whose elements are closed hyper-intervals called cells. Let $\bar{\hat{X}}\subseteq\hat{X}$ be a compact set of congruent hyper-rectangles aligned on a uniform grid parameterized with a quantization parameter $\eta\in(\R^{+})^n$. Each $\hat{x}\in\bar{\hat{X}}$ is given by $c_{\hat{x}}+\segcc{-\frac{\eta}{2},\frac{\eta}{2}}$, where $c_{\hat{x}}\in\eta\Ze^n$ and $\eta\Ze^n=\{c\in\R^n|\exists_{l\in\Ze^n}\forall_{q\in\{1,\dots,n\}}c_q=l_q\eta_q\}$. The cells in $\hat{X}\setminus\bar{\hat{X}}$ are called overflow cells,
    \item $\hat{X}^0=\hat{X}\cap X^0$ and $\hat{U}$ is a finite subset of $U$,
    \item for $\hat{x}\in\bar{\hat{X}}$ and $\hat{u}\in\hat{U}$, a set $A:=\{\hat{x}'\in\hat{X}|\hat{x}'\cap\overline{Reach}(\hat{x},\hat{u})\neq\emptyset$\}. If $A\subseteq\bar{\hat{X}}$ and $\hat{x}'\not\in\hat{X}\setminus\bar{\hat{X}}$, $\forall \hat{x}'\in A$ then $\hat{F}(\hat{x},\hat{u})=A$.
\end{itemize}
\end{definition}
For a detailed procedure on constructing the symbolic model, refer \cite{rungger2016scots}.
\begin{theorem}{\cite[Theorem VIII.4]{ReissigWeberRungger17}}\label{FRRproof}
If $\hat{\Sigma}$ is the symbolic model of $\Sigma$ constructed according to Definition \ref{Def:SymbolicModel} then the relation $Q\subseteq X\times\hat{X}$ defined by $Q=\{(x,\hat{x}) \in X\times \hat{X}: x \in \hat{x}\}$ is feedback refinement relation from $\Sigma$ to $\hat{\Sigma}$.
\end{theorem}


\subsection{Controller Synthesis using Symbolic Models}
Consider the transition system $\Sigma=(X,X^0,U,F)$ and a memoryless controller $C:X\rightrightarrows U$, where for all $x\in X$, $C(x)\subseteq U^a(x)$. Let the domain of the controller be $dom(C):=\{x\in X|C(x)\neq\emptyset\}$.

\begin{definition}\label{controlledsys}
Given a controller $C$ and a transition system $\Sigma$, the controlled system is given by the tuple $\Sigma|C=(X_{C},X^0_C,U_C,F_C)$, where
\begin{itemize}
    \item $X_C=X\cap dom(C)$, $X^0_C\subseteq X_C$, $U_C=U$,
    \item for $x_C\in X_C$ and $u_C\in U_C$, $x'_C\in F_C(x_C,u_C)$ iff $x'_C\in F(x_C,u_C)$ and $u_C\in C(x_C)$.
\end{itemize}
\end{definition}

Let $P$ be the set of atomic propositions that labels the states in $X\subset \R^n$ through a labelling function $\mathcal{L}:X\rightarrow 2^P$ and $\varphi$ be an LTL specification over $P$. 
The control inputs $\mathbf{\mathsf{u}}$ applied to the system according to $C$ generates a trajectory $\mathsf{x}_{x\mathsf{u}}$ from the initial state $x$. We say that the system $\Sigma|C\models\varphi$ if $\mathcal{L}(\mathsf{x}_{x\mathsf{u}})\models\varphi$. For more information on specification satisfaction, refer \cite[Definition VI.1]{ReissigWeberRungger17}.

Given the symbolic model $\hat{\Sigma}$ and the relation $Q\subseteq X\times\hat{X}$, 
we first synthesize a controller $\hat{C}$ such that $\hat{\Sigma}|\hat{C}\models\hat{\varphi}$ using graph theoretical approaches \cite{tabuada2009verification}, 
where $\hat{\varphi}$ is the symbolic specification associated with $\Sigma$, $\hat{\Sigma}$, $Q$ and $\varphi$ such that, if $\mathcal{L}(\hat{x})\subseteq\hat{\varphi}$ 
and $(x,\hat{x})\in Q$, then $\mathcal{L}(x)\subseteq\varphi$.

\begin{theorem}{\cite[Theorem VI.3]{ReissigWeberRungger17}}\label{control}
If $\Sigma\preceq_Q\hat{\Sigma}$ and $\hat{C}$ is the symbolic controller such that $\hat{\Sigma}|\hat{C}\models\hat{\varphi}$, then $\Sigma|C\models\varphi$, where $C:=\hat{C}\circ Q$.
\end{theorem}

The controller $\hat{C}$, synthesized for symbolic model $\hat{\Sigma}$ satisfying global and local specifications, defined in Problem \ref{problem}, can be refined for the concrete system $\Sigma$ with the help of relation $Q$. The controlled system $\Sigma|C\models\varphi$, where $\varphi:=\Psi\wedge\Phi$ is the combination of local and global specifications. Many toolboxes are available for symbolic controller synthesis, for example, \cite{rungger2016scots,omegaThread} and \cite{QUEST}.
For scalability, we first synthesize a controller that satisfies the local specification $\psi_i$ for the $i^{th}$ local symbolic model ($i^{th}$ agent) with $\Psi = \bigwedge_{i\in I}\psi_i$ and enforce global safety specification $\Phi$ on the composed symbolic model with the notion of control barrier functions, discussed in the next section.



\section{Barrier Certificate for Symbolic Models}\label{barriersection}
To deal with global safety specification $\Phi$, we leverage the concept of control barrier function \cite{discreteCBF}. Consider a discrete-time system as defined in \eqref{sysMAS}.
Let the safe set $S\subseteq X\subset \R^n$ be defined as the superlevel set of a continuous function  $B:\R^n\rightarrow\R$ and is given by: 
\begin{align}
    S &= \{x\in X|B(x)\geq 0\}, \label{safeset1}\\
    Int(S) &= \{x\in X|B(x)>0\}, \label{safeset2}\\
    \partial S &=\{x\in X|B(x)=0\}. \label{safeset3}
\end{align}
\begin{definition}
    Given the transition system $\Sigma=(X,X^0,U,F)$ in Definition \ref{Def:trans_sys} and a set $S \subseteq X$. The set $S$ is said to be controlled invariant for the transition system $\Sigma$ if for all $x \in S$, there exists $u \in U$ satisfying $F(x,u) \subseteq S$.
\end{definition}
\begin{theorem}{\cite[Lemma 1]{discreteCBF}}\label{barrierproof}
{Given the discrete-time control system \eqref{sysMAS} and a safe set $S\subseteq X \subset \R^n$ as defined in \eqref{safeset1}-\eqref{safeset3} as a superlevel set of a continuous function $B:\mathbb{R}^n \rightarrow \mathbb{R}$. The set $S$ is controlled invariant for the system (\ref{sysMAS}) if there exists a $\mathcal{K}_{\infty}$ map $\alpha$, satisfying $\alpha(r) < r$, for all $r > 0$ and such that the following holds: for all
$ x\in X\text{, there exists } u\in U$ satisfying
\begin{align}\label{barriertransition}
    B(\mathbf{\mathsf{f}}(x,u))-B(x)\geq -\alpha(B(x)).
\end{align}\label{barrierinvar}
}
\end{theorem}
To solve Problem \ref{problem} in a scalable way, we first define a symbolic safe set $\hat S$ that is compatible with the symbolic model using barrier function $B$ \eqref{safeset1}-\eqref{barriertransition} defining safe set $S$ for the original system \eqref{sysMAS}. For this, we need the following assumption over function $B:X\rightarrow\R$. 

\begin{assumption}\label{A1}
The barrier functions $B:X\rightarrow\R$ defined in Theorem \ref{barrierproof} satisfy the global Lipschitz continuity condition: there exists a constant {$\mathcal{L}^x\in\R^+_0$} such that $\|B(x)-B(y)\|\leq {\mathcal{L}^x}\|x-y\|$ for all $x,y\in X$. 
\end{assumption}


Given Assumption \ref{A1} and a symbolic model $\hat{\Sigma}=(\hat{X},\hat{X}^{0},\hat{U},\hat{F})$ with a symbolic state given by $\hat{x}:=c_{\hat{x}}+\segcc{-\frac{\eta}{2},\frac{\eta}{2}}\in\hat X$ and state space discretization $\eta=(\eta_1,\ldots,\eta_n)$ $\in(\R^+)^n$ as defined in Definition \ref{Def:SymbolicModel}, we define a symbolic safe set $\hat S$ using barrier function $B$ in Theorem \ref{barrierproof} as:
\begin{align}
    \hat S=\{\hat{x}\in\hat{X}|B(c_{\hat{x}})-\mathcal{L}^x\frac{\eta_{max}}{2}\geq0\}\label{safesymbol1},\\
    Int(\hat S)=\{\hat{x}\in\hat{X}|B(c_{\hat{x}})-\mathcal{L}^x\frac{\eta_{max}}{2}>0\}\label{safesymbol2},\\
    \partial \hat S=\{\hat{x}\in\hat{X}|B(c_{\hat{x}})-\mathcal{L}^x\frac{\eta_{max}}{2}=0\},\label{safesymbol3}
\end{align}
where $\eta_{max}=\mathbf{max}_{j\in\{1,\dots,n\}}\eta_j$.

\begin{theorem}\label{mainth}
Consider a system $\Sigma=(X,X^0,U,F)$, its symbolic model $\hat{\Sigma}=(\hat{X},\hat{X}^0,\hat{U},\hat{F})$ constructed with relation $Q\subseteq X\times\hat{X}$ and state space quantization $\eta\in(\R^+)^n$ as defined in Definition \ref{Def:SymbolicModel}, a safe set $S$ as defined in \eqref{safeset1}-\eqref{safeset3}, and the symbolic safe set $\hat S$ as defined in \eqref{safesymbol1}-\eqref{safesymbol3}. If for all ${\hat{x}}\in \hat S$, there exists $\hat{u}\in\hat{U}^a(\hat{x})$ such that 
$$\mathbf{min_{\hat{x}'\in \hat{F}(\hat{x},\hat{u})}}[B(c_{\hat{x}'})-B(c_{\hat{x}})]\geq -\alpha (B(c_{\hat{x}})-\mathcal{L}^x\frac{\eta_{max}}{2}),$$
where $B:X\rightarrow\R$ and $\alpha\in\mathcal{K}_{\infty}$ are defined in Theorem \ref{barrierproof}, then {$Q^{-1}(\hat S)\subset S$} and $\hat{S}$ is invariant for system $\hat\Sigma$. 

\end{theorem}
\begin{proof} Let us first show that $Q^{-1}(\hat S) \subset S$.\\
 From Definition \ref{Def:SymbolicModel}, we have $\hat{x}:=c_{\hat{x}}+\segcc{\frac{-\eta}{2},\frac{\eta}{2}}\in\hat{X}$ and using the fact $\eta_{max}=\mathbf{max}_{j\in\{1,\dots,n\}}\eta_j$, one obtains for all $ (x,\hat{x})\in Q\text{, }\lVert c_{\hat{x}}-x \rVert \leq \frac{\eta_{max}}{2}.$ Using Lipschitz continuity of $B$, $\forall (x,\hat x)\in Q$, we get 
\begin{align*}
B(c_{\hat{x}})-B(x) \leq \lVert B(c_{\hat{x}})-B(x) \rVert \leq \mathcal{L}^x\lVert c_{\hat{x}}-x \rVert \leq \mathcal{L}^x \frac{\eta_{max}}{2}, 
\end{align*}
where $\mathcal{L}^x$ is the Lipschitz constant of the function $B$. Thus,
\begin{align}\label{proofeq}
\forall (x,\hat x)\in Q, \ B(c_{\hat{x}})-\mathcal{L}^x\frac{\eta_{max}}{2}\leq B(x).    
\end{align}
Thus, $B(c_{\hat{x}})-\mathcal{L}^x\frac{\eta_{max}}{2}\geq0\implies B(x)\geq0$, i.e., for all ${\hat{x}}\in \hat S$, $Q^{-1}(\hat{x})\subset S$. 
Thus we have {$Q^{-1}(\hat S)\subset S$}.\\ 
Now to show that $\hat S$ is invariant for $\hat{\Sigma}$, we have that for all $\hat{x}\in\hat S$ there exists $\hat{u}\in\hat{U}^a(\hat{x})$ such that 
\begin{align*}
    &\mathbf{min_{\hat{x}'\in \hat{F}(\hat{x},\hat{u})}}[B(c_{\hat{x}'})-B(c_{\hat{x}})]\\&=\mathbf{min_{\hat{x}'\in \hat{F}(\hat{x},\hat{u})}}[B(c_{\hat{x}'})-\mathcal{L}^x\frac{\eta_{max}}{2}-B(c_{\hat{x}})+\mathcal{L}^x\frac{\eta_{max}}{2}]\\&\geq -\alpha (B(c_{\hat{x}})-\mathcal{L}^x\frac{\eta_{max}}{2}).
\end{align*}
Thus one has for all $\hat{x}\in\hat S$ there exists $\hat{u}\in\hat{U}^a(\hat{x})$ such that
{\small
\begin{align*}
    &B(c_{\hat{x}'})-\mathcal{L}^x\frac{\eta_{max}}{2}-B(c_{\hat{x}})+\mathcal{L}^x\frac{\eta_{max}}{2}\geq -\alpha (B(c_{\hat{x}})-\mathcal{L}^x\frac{\eta_{max}}{2}),
    \\&
    \implies B(c_{\hat{x}'})-\mathcal{L}^x\frac{\eta_{max}}{2}\geq (I_d-\alpha)\circ (B(c_{\hat{x}})-\mathcal{L}^x\frac{\eta_{max}}{2}),
\end{align*}}%
for all ${\hat{x}'}\in\hat F(\hat x,\hat u)$.
Since $\alpha\in\mathcal{K}_{\infty}$ one has that $(I_d-\alpha)\in\mathcal{K}_{\infty}$ which implies from condition \eqref{safesymbol1} (i.e., $B(c_{\hat{x}})-\mathcal{L}^x\frac{\eta_{max}}{2}\geq 0$ for all $\hat x\in \hat 
S$) that for all $\hat x\in\hat S$ we have $B(c_{\hat{x}'})-\mathcal{L}^x\frac{\eta_{max}}{2}\geq 0\implies {\hat{x}'}\in \hat S$, $\forall \hat{x}'\in \hat{F}(\hat{x},\hat{u})$.
This proves the invariance of the set $\hat S$.
\end{proof}

\begin{remark}
Since we know that $Q^{-1}(\hat{S})\subset S$ (from Theorem \ref{mainth}) and with $\hat{S}$ as invariant, the system does not violate the safety specification $\Phi$ by staying inside $S$.

\end{remark}

\section{Scalable Controller Synthesis for Multi-Agent Systems}
\subsection{Controller Synthesis for Each Agent (Symbolic Model)}
Consider the problem of controller synthesis for each agent \eqref{sys} represented by the transition system $\Sigma_i=(X_i,X^0_i,U_i,F_i)$ given a local LTL specification $\psi_i$. Using symbolic control, we first construct the symbolic model of each agent $\Sigma_i$ given by $\hat{\Sigma}_i=(\hat X_i,\hat X_i^0,\hat U_i,\hat F_i)$ (as discussed in Section \ref{symconst}) such that $\Sigma_i\preceq_{Q_i}\hat{\Sigma}_i$, where $Q_i\subseteq X_i\times\hat{X}_i$ is the strict feedback refinement relation. We then synthesize a controller $\hat{C}_i$ such that $\hat{\Sigma}_i|\hat{C}_i\models\hat{\psi}_i$, where $\hat{\psi}_i$ is the symbolic specification for $\hat{\Sigma}_i$ (related to $\psi_i$, $\Sigma_i$ and $Q_i$). 
Theorem \ref{control} shows that we can refine the controller $\hat{C}_i$ using the feedback refinement relation $Q_i$ and the refined controller $C_i:=\hat{C}_i\circ Q_i$ is such that $\Sigma_i|C_i\models\psi_i$.

After controller synthesis, we obtain the controlled agents $\hat{\Sigma}_i|\hat{C}_i=(\hat X_{C_i},\hat X^0_{C_i},\hat U_{C_i},\hat F_{C_i})$, $i \in \{1,2,\ldots,N\}$, as shown in Definition \ref{controlledsys}, where $\hat{X}_{C_i}=\hat{X}_i\cap dom(\hat C_i)$, $\hat{X}^0_{C_i}\subseteq\hat{X}_{C_i}$, $\hat{U}_{C_i}=\hat{U}_i$ and for $\hat{x}\in\hat{X}_{C_i}$, $\hat{u}\in\hat{U}_{C_i}$, $\hat{x}'\in\hat{F}_{C_i}(\hat{x},\hat{u})$ iff $\hat{x}'\in\hat{F}_i(\hat{x},\hat{u})$ and $\hat{u}\in\hat{C}_i(\hat{x})$.

\subsection{Composition using Control Barrier Certificates}
We will now compose the individual symbolic models of the controlled systems of each agent.


Given a collection of $N\text{ }(\in\N)$ controlled systems where each controlled system is given by $\hat{\Sigma}_i|\hat{C}_i=(\hat X_{C_i},\hat X^0_{C_i},\hat U_{C_i},\hat F_{C_i})$ and $I=\{1,\dots,N\}$, the composed controlled system is given by $\hat{\Sigma}|\hat{C}=(\hat{X}_C,\hat{X}^0_C,\hat{U}_C,\hat{F}_C)$ constructed based on Definition \ref{DefComposition}, where $\hat{X}_C=\prod_{i\in I}\hat{X}_{C_i}$, $\hat{X}^0_{C}\subseteq\hat{X}_C$, $\hat{U}_C=\prod_{i\in I}\hat{U}_{C_i}$ and for $\hat{x}\in \hat{X}_C$ and $\hat{u}\in \hat{U}_C$, $\hat F_C(\hat{x},\hat{u})=\prod_{i\in I}\hat{F}_{C_i}(\hat{x},\hat{u})$.

\begin{definition}\label{safetycontroller}
Let $B:X\rightarrow \R$ be the CBF that enforces the safety specification $\Phi$. We construct the safety controller $\hat C_S$ for the system $\hat{\Sigma}|\hat{C}$ defined above, as follows:
\begin{itemize}
    \item $\hat C_S(\hat x)=\{\hat{u}\in\hat{U}^a_C(\hat{x})|\mathbf{min_{\hat{x}'\in \hat{F}(\hat{x},\hat{u})}}[B(c_{\hat{x}'})-B(c_{\hat{x}})]\geq -\alpha (B(c_{\hat{x}})-\mathcal{L}^x\frac{\eta_{max}}{2})\}$ and
    \item $dom(\hat{C}_S)\subseteq \hat{S}\cap\{\hat{x}\in\hat{X}_C|\hat{C}_S(\hat{x})\neq\emptyset\}$.
\end{itemize}
\end{definition}



We can now construct a controlled system $(\hat{\Sigma}|\hat{C})|\hat{C}_S=(\hat X_S,\hat X^0_S,\hat U_S,\hat F_S)$ as defined in Definition \ref{controlledsys}, where $\hat{X}_S=\hat{X}_C\cap dom(\hat C_S)$, $\hat{X}^0_S\subseteq\hat{X}_S$, $\hat{U}_S=\hat{U}_C$ and for $\hat{x}\in \hat{X}_S$ and $\hat{u}\in \hat{U}_S$, $\hat{x}'\in \hat{F}_S(\hat{x},\hat{u})$ iff $\hat{x}'\in\hat{F}_C(\hat{x},\hat{u})$ and $\hat{u}\in\hat{C}_S(\hat{x})$.

\begin{remark}
Note that only the transitions that lead back to the set $\hat S$ are included in the transition system $(\hat \Sigma|\hat C)|\hat C_S$. At some $\hat x\in \hat{X}_S$, $\hat U^a_S(\hat x)$ may be empty because there could be no inputs in $\prod_{i\in I}\hat U^a_{C_i}(\hat x_i)$ that brings the system to $\hat S$.
\end{remark}



 To restore the local specifications (violated due to the safety-enforcing barrier certificate), it is necessary to synthesize a controller $\hat C_B$ for the composed symbolic model's specification given by $\hat{\Psi}=\bigwedge_{i\in I}\hat{\psi}_i$.
 
{The following result shows that the combination of the controllers $ C_B$, $C_S$ and $C_i$, $i\in \{1,2,\ldots,N\}$, designed before makes it possible for the discrete-time control system in (\ref{sysMAS}) to satisfy the control objective defined in Problem \ref{problem}.}

\begin{theorem}
Given the controlled agents $\hat{\Sigma}_i|\hat{C}_i$ with $i\in I=\{1,\dots,N \}$, the strict relation $Q_i\subseteq X_i\times\hat{X}_i$, LTL specification $\Psi:=\bigwedge_{i\in I}\psi_i$, symbolic specification $\hat{\Psi}$ resulting from the concrete specification $\Psi$, the symbolic safe set $\hat{S}$ defined in \eqref{safesymbol1}-\eqref{safesymbol3} and the safe set $S$ defined in \eqref{safeset1}-\eqref{safeset3}, if $\hat{C}_B$ is a controller such that $((\hat \Sigma|\hat C)|\hat C_S)|\hat{C}_B\models\hat{\Psi}$ then, $((\Sigma|C)|C_S)|C_B\models\Psi$, where $C_B:=\hat{C}_B\circ Q$ and the trajectories of the controlled MAS $((\Sigma|C)|C_S)|C_B$ stays inside the safe set $S$.
\end{theorem}

\begin{proof}
From \cite[Theorem VI.3]{ReissigWeberRungger17} and \cite[Corollary VI.5]{ReissigWeberRungger17}, we know that
$
    \Sigma_i|C_i\preceq_{Q_i}\hat{\Sigma}_i|\hat{C}_i
$ since $\Sigma_i\preceq_{Q_i}\hat{\Sigma}_i$ by construction and $C_i:=\hat{C}_i\circ Q_i$.

By composing the controlled agents and since there is no coupling between the agents, one gets
$\Sigma|C\preceq_Q\hat{\Sigma}|\hat{C}$,
where $\Sigma$ is the MAS resulting from the composition of the agents $\Sigma_i$, where $i \in I$ and $I=\{1,\ldots,N\}$, $\hat{\Sigma}$ is the transition system resulting from the composition of the local symbolic models $\hat{\Sigma}_i$, $i \in I$. The set of states for the composed MAS and the composed local symbolic models are given by $X:=\prod_{i\in I}X_{i}$ and $\hat{X}:=\prod_{i\in I}\hat{X}_{i}$, respectively. The controller $C:X \rightrightarrows U$ is defined for $x=(x_1, \ldots, x_N) \in X$ as $u=(u_1, \ldots, u_N) \in C(x)$ if and only if $u_i \in C_i(x_i)$, for all $i \in I$. Similarly the symbolic controller $\hat{C}:\hat{X} \rightrightarrows \hat{U}$ is defined for $\hat{x}=(\hat{x}_1, \ldots, \hat{x}_N) \in \hat{X}$ as $\hat{u}=(\hat{u}_1, \ldots, \hat{u}_N) \in \hat{C}(\hat{x})$ if and only if $\hat{u}_i \in \hat{C}_i(\hat{x}_i)$, for all $i \in I$. The feedback refinement relation $Q\subseteq X\times\hat{X}$ is defined as $Q:=\{(x,\hat{x}) \in X\times\hat{X} \mid (x_i,\hat{x}_i) \in Q_i,\forall i \in I\}.$ 
We synthesize a safety controller $\hat C_S$ as given in Definition \ref{safetycontroller} for the system $\hat{\Sigma}|\hat{C}$. By construction, this controller ensures that the controlled system $(\hat{\Sigma}|\hat{C})|\hat{C}_S$ never leaves $\hat{S}$. The refined controller $C_S:=\hat{C}_S\circ Q$ is such that,
  $  (\Sigma|C)|C_S\preceq_Q(\hat{\Sigma}|\hat{C})|\hat{C}_S
$ since $\Sigma|C\preceq_Q\hat{\Sigma}|\hat{C}$.

We now synthesize a controller $\hat{C}_B$ such that $((\hat{\Sigma}|\hat{C})|\hat{C}_S)|\hat{C}_B\models\hat \Psi$ and since $(\Sigma|C)|C_S\preceq_Q(\hat{\Sigma}|\hat{C})|\hat{C}_S$, $C_B:=\hat{C}_B\circ Q$ is the refined controller such that $((\Sigma|C)|C_S)|C_B\models\Psi$ due to Theorem \ref{control}.

With the composed controlled system we have,
\begin{gather}\label{relation}
    ((\Sigma|C)|C_S)|C_B\preceq_Q((\hat{\Sigma}|\hat{C})|\hat{C}_S)|\hat{C}_B.
\end{gather}

Since $\hat{C}_B(\hat{x})\subseteq\hat{U}_S(\hat{x})$ and $dom(\hat C_B)\subseteq\hat{X}_S$, all trajectories of the system $((\hat{\Sigma}|\hat{C})|\hat{C}_S)|\hat{C}_B$ evolve within the set $\hat S$. From \eqref{relation}, it is clear that $((\Sigma|C)|C_S)|C_B$ will also remain in $\hat{S}$, which implies from Theorem \ref{mainth} that the system $((\Sigma|C)|C_S)|C_B$ always remains in $S$.
Hence, the combination of refined controllers $C$, $C_S$ and $C_B$ allows to satisfy the specifications defined in Problem \ref{problem}.
\end{proof}

\subsection{Conservativeness of the Proposed Approach}
The size and domain of the controller synthesized by the proposed approach are influenced by the class-$\mathcal{K}_\infty$ function $\alpha(r)$. The following theorem provides an analysis of the effect of the function $\alpha$ on the obtained controller.
\begin{figure}
    \centering
    \includegraphics[width=0.45\textwidth]{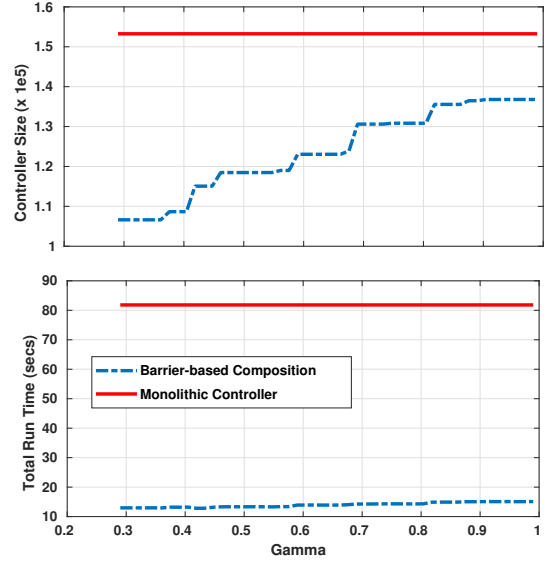}
    \caption{Variation of running time and controller size (measured by the number of allowed transitions) with $\gamma$ compared for proposed and monolithic controller synthesis approach.} 
    \label{fig:gammaInfluence}
\end{figure}

{\begin{theorem}\label{SubsetofContandDom}
{Given the symbolic model of a controlled system $\hat{\Sigma}|\hat{C}=(\hat X_C, \hat X^0_C, \hat U_C, \hat F_C)$, constructed as shown in Definition \ref{Def:SymbolicModel},
global safety specification $\Phi$ expressed in the form of barrier functions as shown in (\ref{safesymbol1}) - (\ref{safesymbol3}) and safety controllers $\hat{C}_{S_1}(\hat x)$ with $\alpha(r):=\gamma_1 r$ and $\hat{C}_{S_2}(\hat x)$ with $\alpha(r):=\gamma_2 r$, $\forall r\in \R^+$ as given in Definition \ref{safetycontroller}, 
where $\gamma_1,\gamma_2\in(0,1)$, if $\gamma_1>\gamma_2$,
then $\hat{C}_{S_2}(\hat x)\subseteq \hat{C}_{S_1}(\hat x)$ {for all $\hat{x} \in \hat{X}_C$} and $dom(\hat{C}_{S_2})\subseteq dom(\hat{C}_{S_1})$.} 
\end{theorem}}

\begin{proof}
{The control barrier function-based controller, as defined in Definition \ref{safetycontroller} with $\alpha(r):=\gamma r$, $\forall r\in\R^+$ after rearrangement of terms is given by: 
\begin{equation}\label{rearrangecontrol}
\begin{split}
    &\hat{C}_S(\hat{x}):=\{\hat u\in\hat{U}^a_C(\hat{x})|\mathbf{min}_{\hat{x}'\in\hat{F}_C(\hat{x},\hat u)}[B(c_{\hat{x}'})]\\&\geq(1-\gamma)B(c_{\hat{x}})+\gamma\mathcal{L}^x\frac{\eta_{max}}{2}\},
\end{split}
\end{equation} for each $\hat{x}\in\hat{S}$. Consider two constants $\gamma_1$, $\gamma_2\in(0,1)$ and {with} $\gamma_1>\gamma_2$. After negating, multiplying with $B(c_{\hat{x}})-\mathcal{L}^x\frac{\eta_{max}}{2}$ and adding {$B(c_{\hat{x}})$} on both sides, 
\begin{equation*}
B(c_{\hat{x}})\hspace{-0.1em}-\gamma_1\hspace{-0.1em}\left(\hspace{-0.2em}B(c_{\hat{x}})\hspace{-0.2em}-\hspace{-0.2em}\mathcal{L}^x\frac{\eta_{max}}{2}\hspace{-0.2em}\right)\hspace{-0.2em}<\hspace{-0.2em} B(c_{\hat{x}})\hspace{-0.1em}-\gamma_2\hspace{-0.1em}\left(\hspace{-0.2em}B(c_{\hat{x}})\hspace{-0.2em}-\hspace{-0.2em}\mathcal{L}^x\frac{\eta_{max}}{2}\hspace{-0.2em}\right)\hspace{-0.2em},
\end{equation*}
\begin{equation}\label{ineqRHS}
    \begin{split}
        {(1\hspace{-0.1em}-\hspace{-0.1em}\gamma_1)B(c_{\hat{x}})\hspace{-0.2em}+\hspace{-0.2em}\gamma_1\mathcal{L}^x\frac{\eta_{max}}{2}} 
        < 
        {(1\hspace{-0.1em}-\gamma_2)B(c_{\hat{x}})\hspace{-0.2em}+\hspace{-0.2em}\gamma_2\mathcal{L}^x\frac{\eta_{max}}{2}}.
    \end{split}
\end{equation}
Let us now consider two sets $E_1$ and $E_2$ as defined below.
\begin{equation*}
\begin{split}
    E_1:=\{\hat{x}'\in\hat{F}_C(\hat{x},\hat u)|\mathbf{min_{\hat{x}'\in\hat{F}_C(\hat{x},\hat u)}}[B(c_{\hat{x}'})]\\\geq(1-\gamma_1)B(c_{\hat{x}})+\gamma_1\mathcal{L}^x\frac{\eta_{max}}{2}\},\\
\end{split}
\end{equation*}
\vspace{0.25em}
\begin{equation*}
\begin{split}
    E_2:=\{\hat{x}'\in\hat{F}_C(\hat{x},\hat u)|\mathbf{min_{\hat{x}'\in\hat{F}_C(\hat{x},\hat u)}}[B(c_{\hat{x}'})]\\\geq(1-\gamma_2)B(c_{\hat{x}})+\gamma_2\mathcal{L}^x\frac{\eta_{max}}{2}\},
\end{split}
\end{equation*} for each $\hat{x}\in\hat{S}$ and ${\forall} u\in\hat{U}^a_C(\hat{x})$.}
{\textbf{Case 1:} If $\hat{x}'\in E_1$ and $\hat{x}'\in E_2$,
\begin{equation*}\label{ineq1}
\begin{split}
    B(c_{\hat{x}'})\geq(1-\gamma_2)B(c_{\hat{x}})+\gamma_2\mathcal{L}^x\frac{\eta_{max}}{2}\\
    >(1-\gamma_1)B(c_{\hat{x}})+\gamma_1\mathcal{L}^x\frac{\eta_{max}}{2}.
\end{split}
\end{equation*}}

\textbf{Case 2:} If $\hat{x}'\in E_1$ and $\hat{x}'\notin E_2$,
\begin{equation*}\label{ineq2}
    \begin{split}
        (1-\gamma_1)B(\hat{x})+\gamma_1\mathcal{L}^x\frac{\eta_{max}}{2}\leq B(c_{\hat{x}'})\\<(1-\gamma_2)B(\hat{x})+\gamma_2\mathcal{L}^x\frac{\eta_{max}}{2}.
    \end{split}
\end{equation*}Both case 1 and case 2 holds true given (\ref{ineqRHS}).

{\textbf{Case 3:} If $\hat{x}'\notin E_1$ and $\hat{x}'\in E_2$,
\begin{equation*}\label{ineq3}
    \begin{split}
        (1-\gamma_2)B(\hat{x})+\gamma_2\mathcal{L}^x\frac{\eta_{max}}{2}\leq B(c_{\hat{x}'})\\<(1-\gamma_1)B(\hat{x})+\gamma_1\mathcal{L}^x\frac{\eta_{max}}{2},
    \end{split}
\end{equation*}
which contradicts (\ref{ineqRHS}). Hence, we conclude that $E_2\subset E_1$.

Consider controllers $\hat{C}_{S_1}$ and $\hat{C}_{S_2}$ constructed as shown in Definition \ref{safetycontroller} with $\alpha(r):=\gamma_1\times r$ and $\alpha(r):=\gamma_2\times r$, $\forall r\in\R^+$ respectively. The controllers $C_{S_1}$ and $C_{S_2}$ can be given as in (\ref{rearrangecontrol}) with $\gamma=\gamma_1$ and $\gamma=\gamma_2$ respectively. Given $E_2\subset E_1$ and (\ref{rearrangecontrol}), $\exists u\in\hat{U}^a_C(\hat{x})$ such that, $u\in C_{S_1}(\hat{x})$ but $u\notin C_{S_2}(\hat{x})$ since the set of $\hat{F}_C(\hat{x},u)$ allowable for $\hat{C}_{S_1}(\hat{x})$, defined by $E_1$, is larger. Thus, we have
    $\hat{C}_{S_1}(\hat{x})\subseteq\hat{C}_{S_2}(\hat{x})$.} From the definition of 
the domain of controller, one has $dom(\hat{C}_{S_2})\subseteq dom(\hat{C}_{S_1})$. This ends the proof.
\end{proof}
\begin{remark}
{The theorem can be generalized to any map $\alpha(r)$, even non-linear ones. From Theorem \ref{barrierproof}, $\alpha(r)<r$, $\forall r>0$. Thus, $\alpha(r)\leq\gamma r$, $\forall r>0$, where $\gamma\in(0,1)$ is the smallest value for which the condition holds. If $\alpha_1(r)>\alpha_2(r)$, $\alpha_1(r)\leq\gamma_1 r$, $\alpha_2(r)\leq\gamma_2 r$, $\forall r>0$ and $\gamma_1,\gamma_2\in(0,1)$ is the smallest value for which the conditions hold, then $\gamma_1>\gamma_2$, and the proof follows as shown.}
\end{remark}
\begin{remark}
The theorem above shows that with an increase in $\gamma$, the size and domain of the safety controller increase. At $\gamma=1$, the inequality is $\mathbf{max_{\hat{x}'\in\hat{F}_C(\hat{x},u)}}B(c_{\hat{x}'})\geq \mathcal{L}^x\frac{\eta_{max}}{2}$ which will lead to all input values $u\in\hat{U}^a_C$ leading to safe states being included in the controller giving us a fully permissible controller for safety. However, this is not possible as $\alpha(r)<r$, $\forall r\in\R^+$ as given in Theorem \ref{barrierproof}, which leads to $\gamma<1$ as $\alpha(r):=\gamma\times r$. This condition is necessary for the {control invariance} of the safe set.
\end{remark}
\begin{figure}[h]
    \centering
    \vspace{0.5em}
    \includegraphics[scale=0.4]{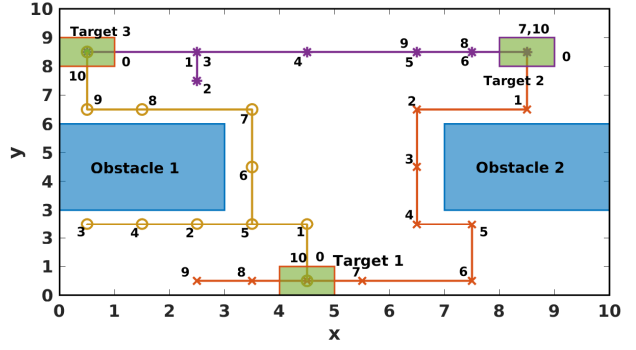}
    \vspace{-1.5em}
    \caption{Simulation for a multi-agent system with three agents}
    \label{fig:Sim}
\end{figure}
\section{Experimental Results}
We compare the proposed approach with a centralized controller synthesis technique. We simulate a discrete system where each agent is given by 
\begin{align}\label{mecanum}
\mathsf{x}_i(k+1)=\mathsf{x}_i(k)+\mathsf{u}_i(k),\text{ }i\in I,\text{ }k\in\N_0,
\end{align}
where $I=\{1,\ldots,N\}$, $\mathsf{x}_i(k)\in X=[0,10]\times[0,10]\subset \R^2$ is the state of the system and $\mathsf{u}_i(k)\in U=\{(-2,0),(-1,0),(1,0),(2,0),(0,-2),(0,-1),(0,1),(0,2)\}\\\subset \R^2$ is the input to the system. We ran two experiments with $N=2$ and $N=3$.
The global safety specification is given by a set of pair-wise CBFs as
\begin{gather}\label{globalspec}
    B_{ij}(x)=\lVert x_i-x_j\rVert-d_{ij},\text{ }\forall i,j\in I\text{ and }i\neq j,
\end{gather} where $d_{ij}=3$ is the desired distance between the agents $i$ and $j$. The local LTL specification for agent $i$ is given by $\psi_i=\lozenge Target_i\wedge(\Box\neg (Obs_1\lor Obs_2))$, where $Target_i$ is the target of agent $i$, $Obs_1$ and $Obs_2$ are the obstacles in the state-space, $\Box$ and $\lozenge $ represent temporal operators always and eventually, respectively. We used a computer with AMD Ryzen 9 5950x, 128 GB RAM, and NVIDIA RTX 3080Ti graphics card to perform simulations that were run on MATLAB. The state quantization parameter is $\eta=[1,1]$. 

\begin{table}[]
\centering
\caption{Computation time comparison}
\label{compare}
\begin{tabular}{@{}clll@{}}
\toprule
Number of & \multicolumn{3}{c}{Computation Time (secs)}         \\ \cmidrule(l){2-4} 
Robots    & Monolithic           & Proposed & Percent reduction \\ \midrule
2         & 170.24               & 32.11    & 81.25             \\
3         & \textgreater 4 weeks & 85985.05 & \textgreater 96.5 \\ \bottomrule
\end{tabular}
\end{table}
{Table \ref{compare} shows the synthesis time of the proposed technique compared to that of the monolithic approach. For the $2$ agent example, the proposed technique took $32.082$ secs for controller synthesis, $81.25\%$ lower than the $170.2421$ secs taken by the monolithic approach. For $3$ agents, the proposed technique took  $85985.05$ secs for synthesis, while the monolithic approach did not finish synthesizing after $4$ weeks, raising the reduction in time to more than $96.5\%$. The proposed technique gets progressively faster compared to the monolithic approach with increased state-space dimensions.}

{We synthesize the controller using our approach for different values of $\gamma$ for $\alpha(r)=\gamma r$ in Definition \ref{safetycontroller}. As indicated by Theorem \ref{SubsetofContandDom}, the number of transitions allowed by the controller increased with higher $\gamma$ values, as illustrated in Figure \ref{fig:gammaInfluence}(top). We also depicted a $\gamma$ vs synthesis time graph in Figure \ref{fig:gammaInfluence}(bottom), demonstrating that increasing $\gamma$ led to a gradual rise in synthesis time. Consequently, we can obtain a less conservative controller with a high $\gamma$ value without significantly compromising synthesis time.} 

Figure \ref{fig:Sim} shows the simulation of a three-agent system with a local specification of avoiding the obstacles, Obs $1$ and Obs $2$, and reaching the corresponding targets while ensuring the global safety specification. 
One can see that the MAS satisfies both local and global specifications. The orange line shows the trajectory of agent $1$; the purple line shows the trajectory of agent $2$; and the yellow line shows the trajectory of agent $3$. 
Figure \ref{fig:Distance} shows the distance between the agents as they move in the arena. The grey horizontal line shows the lower bound on the distance between agents. The graph clearly shows that all three agents never get closer than three units to each other, satisfying the global specification. 

We have implemented our algorithm in a real-world setup comprising of two heterogeneous agents: Mecanum drive, given by \eqref{mecanum}, and Omni-directional, given by $\dot x = u\cos{\theta}+v\sin{\theta}, \dot y=u\sin{\theta}-v\cos{\theta}, \dot\theta=\omega$,
where $(x, y) \text{ and } \theta$ are the robot's position and orientation respectively. The control inputs, $u$, $v$, and $\omega$ are linear velocities and angular velocity, respectively. The demonstrations can be found here \href{https://youtu.be/GtmgvIhH93o}{https://youtu.be/GtmgvIhH93o}.


\begin{figure}[t]
    \centering
    \includegraphics[scale=0.45]{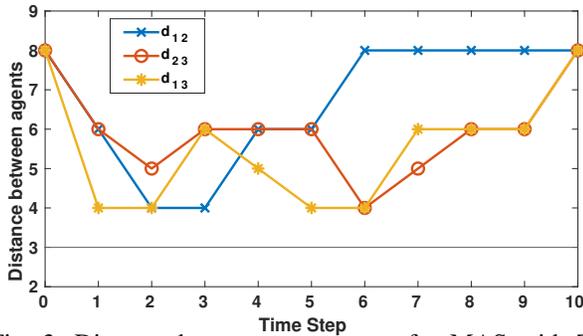}
    \vspace{-1.0em}
    \caption{Distance between two agents for MAS with Three Agents. The grey line is lower bound on the safe distance.} 
    \vspace{-0.75em}
    \label{fig:Distance}
\end{figure}



\section{Conclusion}
{We proposed a three-step bottom-up symbolic approach for MAS by combining symbolic control and CBFs. We have proven the correctness of our approach and have formally shown that the final controlled multi-agent system satisfies both local LTL and global safety specifications. We also formally analysed the conservatism of our barrier function-based approach by analyzing the effect of the control barrier function parameters on the size and domain of the controller. We have also demonstrated the influence of these parameters on synthesis time and implemented our controller synthesis algorithm in a simulation as well as a real-world heterogeneous system, which shows the practicality of our approach.}





\bibliographystyle{ieeetr} 
\bibliography{sources} 

\begin{thebibliography}{10}

\bibitem{tabuada2009verification}
P.~Tabuada, {\em Verification and control of hybrid systems: a symbolic approach}.
\newblock Springer Science and Business Media, 2009.

\bibitem{ReissigWeberRungger17}
G.~Reissig, A.~Weber, and M.~Rungger, ``Feedback refinement relations for the synthesis of symbolic controllers,'' {\em IEEE Transactions on Automatic Control}, vol.~62, no.~4, pp.~1781--1796, 2017.

\bibitem{GAMARA}
R.~Majumdar, K.~Mallik, M.~Salamati, S.~Soudjani, and M.~Zareian, ``Symbolic reach-avoid control of multi-agent systems,'' in {\em Proceedings of the ACM/IEEE 12th International Conference on Cyber-Physical Systems}, pp.~209--220, 2021.

\bibitem{autoSynth}
J.~Dai and H.~Lin, ``Automatic synthesis of cooperative multi-agent systems,'' in {\em 53rd IEEE Conference on Decision and Control}, pp.~6173--6178, 2014.

\bibitem{MDP}
A.~Nikou, J.~Tumova, and D.~V. Dimarogonas, ``Probabilistic plan synthesis for coupled multi-agent systems.,'' {\em IFAC-PapersOnLine}, vol.~50, no.~1, pp.~10766--10771, 2017.
\newblock 20th IFAC World Congress.

\bibitem{Auction}
P.~Schillinger, M.~Bürger, and D.~V. Dimarogonas, ``Auctioning over probabilistic options for temporal logic-based multi-robot cooperation under uncertainty,'' in {\em IEEE International Conference on Robotics and Automation (ICRA)}, pp.~7330--7337, 2018.

\bibitem{swarmreactive}
J.~Chen, R.~Sun, and H.~Kress-Gazit, ``Distributed control of robotic swarms from reactive high-level specifications,'' in {\em IEEE 17th International Conference on Automation Science and Engineering (CASE)}, pp.~1247--1254, 2021.

\bibitem{co-opCoupled}
A.~Nikou, D.~Boskos, J.~Tumova, and D.~V. Dimarogonas, ``Cooperative planning for coupled multi-agent systems under timed temporal specifications,'' in {\em 2017 American Control Conference (ACC)}, pp.~1847--1852, 2017.

\bibitem{MITL}
S.~Andersson, A.~Nikou, and D.~V. Dimarogonas, ``Control synthesis for multi-agent systems under metric interval temporal logic specifications.,'' {\em IFAC-PapersOnLine}, vol.~50, no.~1, pp.~2397--2402, 2017.
\newblock 20th IFAC World Congress.

\bibitem{saoud2021compositional}
A.~Saoud, P.~Jagtap, M.~Zamani, and A.~Girard, ``Compositional abstraction-based synthesis for interconnected systems: An approximate composition approach,'' {\em IEEE Transactions on Control of Network Systems}, vol.~8, no.~2, pp.~702--712, 2021.

\bibitem{decoupled}
R.~Alur, S.~Moarref, and U.~Topcu, ``Compositional and symbolic synthesis of reactive controllers for multi-agent systems,'' {\em Information and Computation}, vol.~261, pp.~616--633, 2018.
\newblock 4th International Workshop on Strategic Reasoning.

\bibitem{multi-barrier}
M.~Srinivasan, S.~Coogan, and M.~Egerstedt, ``Control of multi-agent systems with finite time control barrier certificates and temporal logic,'' in {\em IEEE Conference on Decision and Control (CDC)}, pp.~1991--1996, 2018.

\bibitem{safetyOnly}
L.~Wang, A.~D. Ames, and M.~Egerstedt, ``Safety barrier certificates for collisions-free multirobot systems,'' {\em IEEE Transactions on Robotics}, vol.~33, no.~3, pp.~661--674, 2017.

\bibitem{acASR+barrier}
M.~Mizoguchi and T.~Ushio, ``Abstraction-based symbolic control barrier functions for safety-critical embedded systems,'' {\em IEEE Control Systems Letters}, vol.~6, pp.~1436--1441, 2022.

\bibitem{barrierabstract}
P.~Nilsson and A.~D. Ames, ``Barrier functions: Bridging the gap between planning from specifications and safety-critical control,'' in {\em IEEE Conference on Decision and Control (CDC)}, pp.~765--772, 2018.

\bibitem{verifiable}
J.~Chen, S.~Moarref, and H.~Kress-Gazit, ``Verifiable control of robotic swarm from high-level specifications,'' in {\em Proceedings of the 17th International Conference on Autonomous Agents and MultiAgent Systems}, AAMAS '18, (Richland, SC), p.~568–576, 2018.

\bibitem{sundarsingh2023safeMulti}
D.~S. Sundarsingh, J.~Bhagiya, Saharsh, J.~Chatrola, A.~Saoud, and P.~Jagtap, ``Scalable distributed controller synthesis for multi-agent systems using barrier functions and symbolic control,'' {\em IEEE 62nd Conference on Decision and Control (CDC)}, 2023.

\bibitem{reachable1}
A.~A. Kurzhanskiy and P.~Varaiya, ``Reach set computation and control synthesis for discrete-time dynamical systems with disturbances,'' {\em Automatica}, vol.~47, no.~7, pp.~1414--1426, 2011.

\bibitem{reachIntro}
O.~Maler, ``Computing reachable sets: An introduction,'' 2008.

\bibitem{LTLcite}
C.~Baier and J.-P. Katoen, {\em Principles of Model Checking}.
\newblock The MIT Press, 2008.

\bibitem{barrierfunc}
A.~D. Ames, S.~Coogan, M.~Egerstedt, G.~Notomista, K.~Sreenath, and P.~Tabuada, ``Control barrier functions: Theory and applications,'' in {\em 18th European Control Conference (ECC)}, pp.~3420--3431, 2019.

\bibitem{rungger2016scots}
M.~Rungger and M.~Zamani, ``{SCOTS: A} tool for the synthesis of symbolic controllers,'' in {\em Proceedings of the 19th international conference on hybrid systems: Computation and control}, pp.~99--104, 2016.

\bibitem{omegaThread}
M.~Khaled and M.~Zamani, ``Omegathreads: Symbolic controller design for $\omega$-regular objectives,'' in {\em Proceedings of the 24th International Conference on Hybrid Systems: Computation and Control}, HSCC '21, (New York, NY, USA), Association for Computing Machinery, 2021.

\bibitem{QUEST}
P.~Jagtap and M.~Zamani, ``{QUEST: A} tool for state-space quantization-free synthesis of symbolic controllers,'' in {\em Quantitative Evaluation of Systems} (N.~Bertrand and L.~Bortolussi, eds.), (Cham), pp.~309--313, Springer International Publishing, 2017.

\bibitem{discreteCBF}
Y.~Xiong, D.-H. Zhai, M.~Tavakoli, and Y.~Xia, ``Discrete-time control barrier function: High-order case and adaptive case,'' {\em IEEE Transactions on Cybernetics}, pp.~1--9, 2022.

\end{thebibliography}

\end{document}